\documentclass[acmsmall, manuscript]{acmart}
\usepackage{dsfont}
\usepackage[colorinlistoftodos]{todonotes}

\newtheorem{prop}{Proposition}

\AtBeginDocument{%
  \providecommand\BibTeX{{%
    \normalfont B\kern-0.5em{\scshape i\kern-0.25em b}\kern-0.8em\TeX}}}

\setcopyright{none}

\acmConference[FAccTRec 2022]{FAccTRec 2022}{}{}

\begin{document}

\title[Random Isn't Always Fair: Candidate Set Imbalance and Exposure Inequality]{Random Isn't Always Fair: Candidate Set Imbalance and Exposure Inequality in Recommender Systems}

\author{Amanda Bower}
\author{Kristian Lum}
\author{Tomo Lazovich}
\author{Kyra Yee}
\author{Luca Belli}
\affiliation{%
  \institution{Twitter}
  \city{San Francisco}
  \state{California}
  \country{USA}
}

\begin{abstract}



Traditionally, recommender systems operate by returning a user a set of items, ranked in order of estimated relevance to that user. In recent years, methods relying on stochastic ordering have been developed to create ``fairer" rankings that reduce inequality in who or what is shown to users. Complete randomization--ordering candidate items randomly, independent of estimated relevance--is largely considered a baseline procedure that results in the most equal distribution of exposure. In industry settings, recommender systems often operate via a two-step process in which candidate items are first produced using computationally inexpensive methods and then a full ranking model is applied only to those candidates. In this paper, we consider the effects of inequality at the first step and show that, paradoxically, complete randomization at the second step can result in a higher degree of inequality relative to deterministic ordering of items by estimated relevance scores.  
In light of this observation, we then propose a simple post-processing algorithm in pursuit of reducing exposure inequality that works both when candidate sets have a high level of imbalance and when they do not. The efficacy of our method is illustrated on both simulated data and a common benchmark data set used in studying fairness in recommender systems. 
\end{abstract}




\maketitle

\section{Introduction}
Recommender and ranking systems organize vast troves of information to help people make decisions, like what to buy or who to hire. They also help people find relevant information, like content produced on social media or web pages. 
These recommender systems have historically followed the Probability Ranking Principle \cite{robertson1977probability}: to maximize utility to a user of the system, recommended items should appear in decreasing order of their probability of relevance to the user. However, this approach completely ignores who or what is recommended despite the economic and social impacts to them. For example, a recommender system that does not return an individual's content can preclude that content creators' access to the $\geq$ \$10 billion creator economy \cite{stripe}. Image searches that only return the images with the highest relevance can perpetuate occupational gender stereotypes in images by only, for example, returning images of men when queried for professional roles traditionally associated with men. \cite{kay2015unequal}. This over-indexing on the consumers of the system can lead to extreme inequality in who or what gets exposure. 

These inequalities arise for reasons related to the algorithm itself and how the results of the recommendation algorithm are displayed and interacted with by users. On the non-algorithmic end, due to user interface design choices such as limited recommendation spots and human behavior such as limited attention or preference for highly ranked items regardless of relevance \cite{collins2018study,granka2004eye}, a relatively small number of items are recommended and even fewer engaged with. On the algorithmic end, the Probability Ranking Principle exacerbates the disparities of outcomes between items that are \textit{predicted} to be better than others since they will always be ranked in a better position even when differences in predicted relevance is negligible. In the worst case, one item can receive high levels of exposure while a similar item with only marginally worse estimated relevance can receive none. Because the estimated relevance itself can  have high statistical uncertainty, it is possible that this ``winners-take-all"  dynamic could allocate much more exposure to an item that is, in reality, less relevant. 


Stochastically ranking the items--instead of deterministically ranking them--has emerged as one standard approach to mitigate exposure inequality \cite{singh2019Policy, diaz2020evaluating, wu2022joint, bower2021individually, oosterhuis2021computationally, do2021two, singh2018Fairness, usunier2022fast, kuhlman2021measuring, gorantla2022sampling, garcia2021maxmin, singh2021fairness, heuss2022fairness, pandey2005shuffling}. Stochastic rankings provide flexibility in how expected exposure is allocated to items in comparison to the rigid allocation of exposure under deterministic rankings. For example, one popular approach samples rankings from the Plackett-Luce distribution \cite{singh2019Policy, diaz2020evaluating, wu2022joint, bower2021individually, oosterhuis2021computationally}. Under this sampling approach, for each ranking request, items with higher estimated relevance are more likely to appear lower in the ranking, though there is randomness in the exact order in which items occur. In contrast to the deterministic ranking setting, under the Plackett-Luce  sampling, items with similar predicted relevance will receive similar amounts of exposure in expectation. Another common approach used as a baseline for comparison of fair ranking models is to order the candidate items \cite{diaz2020evaluating, wu2022joint} uniformly at random , i.e. every item has equal probability of appearing in each position. This approach is often described as leading to ``fairer" ranking models that more equally distribute exposure across producers. However, in practice we have found that a pure randomization approach can actually result in an {\it increase} of exposure inequality. 

In this paper, we explore how randomized ranking can increase inequality. In short, ranking in industry settings typically operates in a two-step process since it would be infeasible to rank every single possible item for each user \cite{wang2022fairness}. In the first stage, a computationally inexpensive algorithm generates a relatively small candidate set, and then a more computationally expensive approach ranks the candidates. The phenomenon we identified in which randomization across the candidates increases inequality can occur when candidate sets are imbalanced, i.e. when some items are much more likely to appear in the candidate sets than others. 

While there are several aspects to consider when building responsible recommender systems, like demographic biases in model performance, in this paper we specifically focus on inequality of exposure of the items being recommended. We offer two main contributions. First, we demonstrate that when striving to reduce exposure inequality in recommender systems, candidate set imbalance cannot be ignored. In particular, despite the widespread assumption that randomizing reduces inequality, we demonstrate that for recommender systems with imbalanced candidate sets, this need not be the case. Second, we propose a modification to the class of Plackett-Luce distributions that have been used as a post-processing step to reduce inequality of exposure on the individual ranking level while balancing consumer-utility \cite{diaz2020evaluating, wu2022joint}. This modification uses the system level information of how often an item appears in the candidate sets to reduce exposure inequality. Similar to Plackett-Luce post-processing, our approach is simple as it does not require re-training the candidate generation nor ranking model. Moreover, since post-processing methods like ours operate on the relevance scores regardless of how they were learned or obtained, our approach is appropriate in industry settings where the outputs of several models may be combined to produce a ranking of candidates. We illustrate the effectiveness of this approach on synthetic and real data. 

\section{Related Work} 

Over the past few years, the fair ranking literature has rapidly grown. See \cite{ekstrand2022fairness, patro2022fair, li2022fairness, zehlike2021fairness, sonboli2022multisided} for survey papers. While there are many ethical aspects to consider when building recommender systems, in this work, we are interested in the popular problem of inequality of exposure allocation to the items being recommended. In particular, we are interested in the relationship between equity of exposure at the individual ranking level and at the system level. Exposure at the individual ranking level comes from the exposure items receive in a single ranking request from a user at a single point in time. Exposure at the system level comes from the cumulative exposure items receive over all users and all ranking requests over a time period.

There are several papers that focus on distributing exposure more equitably on the \textit{individual ranking level} \cite{singh2018Fairness, singh2019Policy, diaz2020evaluating, heuss2022fairness, wu2022joint} 
which do not take candidate set inequality into account. For example, the canonical work in \cite{singh2018Fairness} requires that the ratio of exposure that an item receives to its relevance should be constant over all items for each ranking request. One reason to focus on exposure at the individual ranking level is that the associated optimization problems become easier to solve since the dependent nature of system level exposure can be ignored. In light of candidate set imbalance, as we show in our experiments, these type of approaches do not necessarily guarantee equity of exposure at the system level. 
 
On the other hand, there are group-fairness notions of equality of exposure on the individual ranking level where exposure equality does generalize to the system level. In particular, since these approaches put constraints on the proportions of items in the ``top $k$" belonging to each demographic group for each ranking, each group at the system level will receive the same level of exposure as they do in each individual ranking \cite{StoyanovichAndYang, zehlike2017fa, DBLP:journals/corr/CelisSV17, geyik2019fairness, 10.1145/3351095.3372858, 0d8212cdb710458e9a32312c0ffe3a11}. However, these approaches do not naturally extend to individual fairness, the main case we consider. 

Another line of work focuses on distributing and measuring exposure more equitably at the \textit{system level}, which implicitly takes candidate set imbalance into account \cite{surer2018multistakeholder, biega2018Equity, do2021two, do2022optimizing, lazovich2021measuring, zhu2016attention}. Instead of turning to stochastic rankings, the work in \cite{biega2018Equity} amortizes equity of exposure on the system level over time. However, their approach as well as \cite{surer2018multistakeholder} requires solving an integer linear program. 

Recently, the works in \cite{do2021two, do2022optimizing} overcame the previously mentioned optimization challenges that arise when reducing system level exposure directly via the Franke-Wolfe algorithm and leveraging tools in non-smooth optimization. Their approach learns a stochastic ranking model given relevance scores by optimizing welfare functions that take into account user utility and item side inequality. Since our algorithm does not require training a model, our approach can more quickly adapt to dynamic environments that effect candidate set imbalance although recently a fast online algorithm was proposed \cite{usunier2022fast} applicable to differentiable models like in \cite{do2021two} but not \cite{do2022optimizing}. Furthermore, along the way, the authors of \cite{do2021two} also note a similar finding as ours wherein a class of stochastic rankings theoretically behave unexpectedly in terms of producer inequality and consumer utility, albeit in a different setting.  


In either case, stochastic rankings are one standard way of achieving system level or individual-level equity of exposure \cite{singh2019Policy, diaz2020evaluating, wu2022joint, bower2021individually, oosterhuis2021computationally, do2021two, singh2018Fairness, usunier2022fast, kuhlman2021measuring, gorantla2022sampling, garcia2021maxmin, singh2021fairness, heuss2022fairness, pandey2005shuffling}. 
The work in \cite{diaz2020evaluating, wu2022joint} proposes sampling rankings from a class of Plackett-Luce distributions as a post-processing step to achieve individual-level equity of exposure. They make the observation that as stochasticity is increased, inequality of exposure decreases. In the case of a random ranking then, expected exposure is equalized for all items in a single ranking request. However, in our work, due to candidate set inequality, we argue that this post-processing algorithm does not guarantee system level equality of exposure. We make a simple modification to this algorithm \cite{diaz2020evaluating, wu2022joint} by leveraging system level information. 

Finally, very recently, methods to change the candidate sets themselves have emerged. For example, the authors in \cite{wang2022fairness} propose algorithms to obtain group-fair candidate sets. Instead of modifying the candidate sets, our work leverages information about the candidate sets to achieve more equitable rankings in cases where candidate set inequality is significant but not insurmountable. 

\section{Motivating Toy Example}\label{sec:toy}

In this section, we illustrate how system level inequality of exposure can actually increase even when the final ranking a user sees is drawn uniformly at random from all permutations of their candidate set.

To illustrate how completely randomizing each ranking request can increase system level inequality, we present the following toy example. Consider a case where we have consumers \{1, 2, ..., 10\} and producers \{A, B, C, ..., J\}.
For each consumer, the candidate generation process returns four candidates, which are ranked as shown in Figure \ref{fig:toy}. Suppose for simplicity that each consumer stops after the first item that is returned.
\begin{figure}
    \centering
    \includegraphics[width=6in]{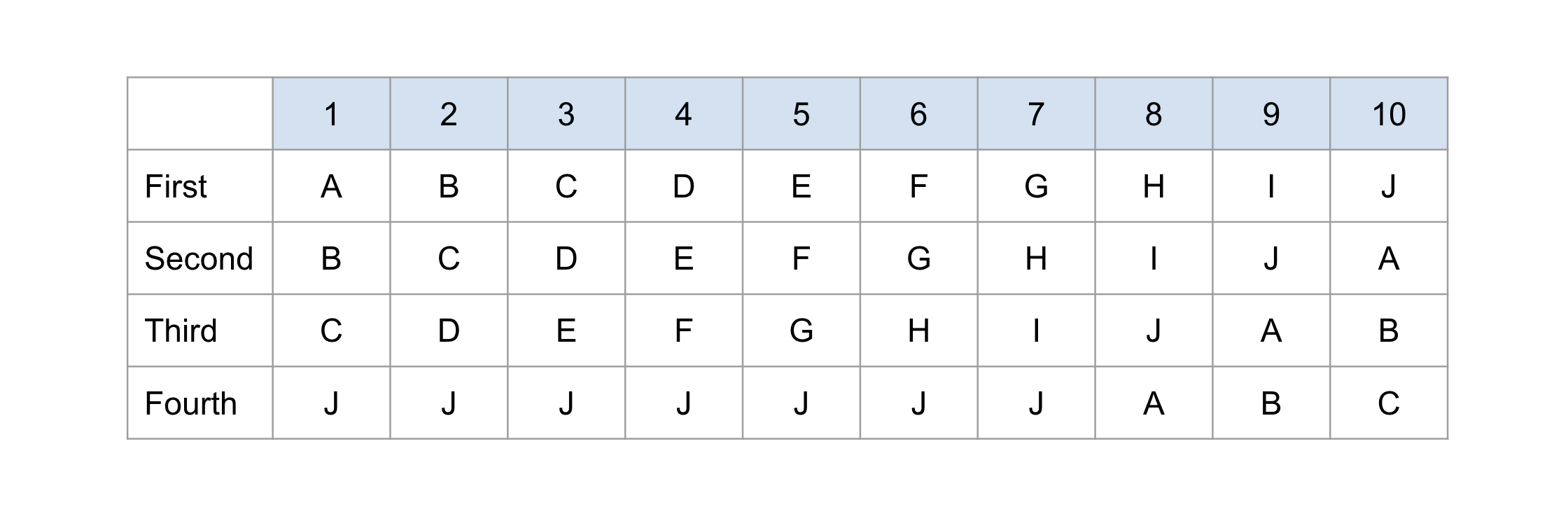}
    \caption{Ranked candidate sets for consumers \{1, 2, ..., 10\} and producers \{A, B, ..., J\}. }
    \label{fig:toy}
\end{figure}
Then, if we use deterministic ranking, every producer receives exactly the same number of impressions--exactly one. However, if we first randomize the order among the four candidates, then producer J has a 1/4 chance of being shown to every consumer, resulting in 2.5 impressions on expectation overall. The expected number of impressions are shown in Figure \ref{fig:expected_impressions} under both scenarios. In this case, it is clear that randomization would increase exposure inequality relative to ranking because of the imbalanced nature of the candidate sets. 

\begin{figure}
    \centering
    \includegraphics[width=6in]{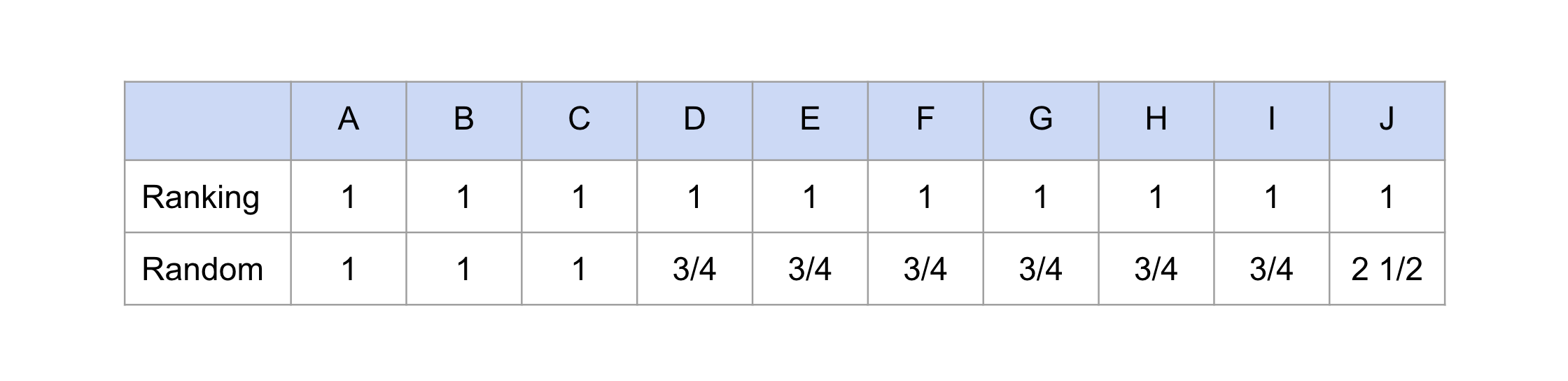}
    \caption{Expected number of impressions for each producer under ranking and randomization.}
    \label{fig:expected_impressions}
\end{figure}

While this example may seem contrived, this phenomenon, where lower ranked content tends to appear much more often in candidate sets overall, can occur in practice. One such reason is to prevent empty or small candidate sets. When the candidate sets that are returned by the first stage model are small, popular items may be used to backfill  candidate sets. These popular items may not be the most relevant items to a user, resulting in low relevance scores, but could appear in many candidate sets to ensure every candidate set is of an appropriate size.

\section{Post-Processing Algorithm: Plackett-Luce With Inverse Candidate Frequency Weights}\label{sec:algo}
We now present a post-processing algorithm to mitigate the effect of candidate set inequality. While it is reasonable to directly intervene on the first stage of candidate generation like in \cite{wang2022fairness}, we argue it is also valid to intervene directly at the second stage of ranking the candidates for a couple reasons. First, candidate sets are populated from many different sources, like popular accounts, accounts that are followed by accounts a user follows, etc. If we directly intervene at the candidate generation stage, we need to take into account dependencies between these multiple sources that would otherwise be independent from one another, which may increase latency to an unacceptable level and may result in small candidate sets for some users since some sources like popular account sources are used to ensure each user has a reasonable number of candidates. Second, intervening on the second step arguably has more of a direct impact on what recommendation a user ultimately sees as there may be hundreds of candidates but only a handful of candidates actually selected to be recommended.

In \cite{diaz2020evaluating, wu2022joint}, the authors propose a simple post-processing algorithm to achieve equity of exposure on the individual ranking level by taking the relevance scores output by a model, transforming the scores either by exponentiation by a constant or multiplying them by a constant, and then sampling a ranking from the Plackett-Luce distribution. Their transformation of scores allows them to smoothly interpolate between completely randomizing the candidate sets (constant is zero) and deterministically ranking the items in decreasing order of relevance score (constant is large). 

We propose a modification to this post-processing algorithm by incorporating the frequency that each item appears in candidate sets in pursuit of achieving system level equity of exposure. We call our approach ``Plackett-Luce With Inverse Candidate Frequency Weights." In particular, suppose there are $n$ users $U := \{u_i\}_{i=1}^n$ and $m$ items $V: = \{v_i\}_{i =1}^{m}$. For every user $u \in U$, let $\{v_{u_i}\}_{i = 1}^{m_{u}}$ be the set of $m_{u}$ candidate items that need to be ranked for user $u$. Let $\{r_{{u}_i}\}_{i = 1}^{m_{u}}$ be the set of relevance scores corresponding to the items. For simplicity, we assume each user requests just one ranking. Let $\alpha, \beta \in \mathbb{R}$. Let $W_{v}$ be the number of candidate sets that item $v \in V$ appears in. Then for user $u \in U$, the probability that the ranking $(v_{u_{1}}, v_{u_{2}}, \dots, v_{m_u})$ is sampled is

\begin{align}\label{eq:pl}
    \prod_{j = 1}^{m_u} \frac{\alpha \exp(\beta r_{u_j}) + \frac{1}{W_{v_{u_j}} + \alpha}}{\sum_{\ell = j}^n \alpha \exp(\beta r_{u_{\ell}}) + \frac{1}{W_{v_
    \ell}+ \alpha}}.
\end{align}

Considering Equation \eqref{eq:pl} for all possible rankings, it is easy to see that it is more likely to sample rankings such that items corresponding to larger, respectively smaller, $\alpha \exp(\beta r_{u_j}) + \frac{1}{W_{v_{u_j}} + \alpha}$ are highly, respectively lowly, ranked. When $\beta = O( \alpha)$ and $\beta >0$, the hyperparameter $\alpha$ allows this distribution to smoothly interpolate between the deterministic ranking given by sorting the items by decreasing relevance score (when $\alpha \rightarrow \infty$) and stochastic rankings that are more likely to highly (lowly) rank items that appear less (more) frequently among all candidate sets (when $\alpha \rightarrow 0^{+}$). When each item appears in the candidate sets at equal rates and $\alpha$ is large, our algorithm behaves like the algorithm in \cite{diaz2020evaluating, wu2022joint}.  For a fixed $\alpha$, $\beta$ controls the level of stochasticity. To efficiently sample rankings from the Plackett-Luce distribution, we utilize ``Algorithm-A" in \cite{efraimidis2006weighted}, which is an $O(N \log(N))$ algorithm that solves the ``weighted sampling" problem when there are $N$ items to rank.

Furthermore, we can easily extend our individual fairness type method to group fairness by replacing $W_{v_i}$, which is the number of times user $v_i$ appears in the candidate sets, with $W_{g(v_i)}$ defined as the number of times a candidate from group $g(v_i)$ appears in the candidate sets where $g(v_i)$ is the demographic group user $v_i$ belongs to. 


\section{Experimental Results}
We consider experiments on a synthetic data set where ground truth relevance scores are known as well as the German credit data set \cite{Dua:german} where relevance scores need to be learned from features of the candidates to test our approach under a more realistic scenario. The German credit data set is a real binary classification data set that is commonly repurposed for ranking tasks in the fair ranking literature \cite{zehlike2021fairness}.  

\subsection{Set-up}
We now describe the common set-up in both experiments. Motivated by the example in Section \ref{sec:toy}, we simulate scenarios where some candidates who have relatively low relevance scores are significantly more likely to appear in a candidate set than others. This can be achieved if the rank at which these candidates appear in the final recommendation is negatively correlated with their scores.

\subsubsection{Basics} Recall that we consider a two-stage process for recommendations. First, for each user requesting a ranking, a candidate set is selected. Then the items in the candidate set are ranked. We assume that the size of each candidate set, $k$, is constant over all users and that each user views exactly the first $\ell$ candidates per ranking presented to them. We assume each user makes exactly one ranking request. Each candidate has a relevance score that is constant over the users. These scores are either given as ground truth in the synthetic experiments or learned from a model in the German credit experiments.

\subsubsection{Candidate Set Construction}\label{sec:candidates}

The candidate sets are sampled from a certain probability distribution over all candidate sets of size $k$. In particular, each candidate $i$ is assigned a candidate score $c_i$. For all but 10 of the $m$ candidates, the candidate score for each candidate is sampled from a beta distribution $B(\alpha = 1, \beta = 10)$. The other 10 are assigned a relatively large candidate score of 5 compared to the scores of the other items so that these items will be more likely to be selected in the candidate sets. Then we build each candidate set iteratively by sampling one candidate out of all possible remaining candidates without replacement until we have selected $k$ candidates. The chance that a candidate is selected at any step is the ratio of its candidate score to the sum of the candidate scores of all remaining potential candidates. 

\subsubsection{Algorithms}
We evaluate our post-processing algorithm \textbf{``Plackett-Luce with Inverse Candidate Frequency Weights"} from Section \ref{sec:algo} against a handful of other baselines. We refer to it as \textbf{``PL-ICFW"} for short. We call the special case of $\alpha, \beta = 0$ in our post-processing algorithm \textbf{``Inverse Weighted."} The algorithm that ranks candidates in decreasing order of relevance score is called \textbf{``Deterministic."} As we stated in the introduction, this type of algorithm has historically been used in recommender systems based on the Probability Ranking Principle \cite{robertson1977probability}. We also compare our approach to sampling rankings from the Plackett-Luce distribution as a function of the relevance scores multiplied by a constant as proposed in \cite{diaz2020evaluating, wu2022joint}. 
We refer to this approach as \textbf{``Scaled Plackett-Luce"} or \textbf{``Scaled PL"} for short. Our approach is similar except we include the extra $\frac{1}{W_{v_{u_j}} + \alpha}$ terms in Equation \eqref{eq:pl} and our hyperparameter $\alpha$. The algorithm that uniformly at random samples a ranking of the candidates is called \textbf{``Randomized"}, which is a special case of Scaled Plackett-Luce. In the German credit experiments, we also compare to the stochastic rankings learned via policy-gradients with individual fairness constraints in \cite{singh2019Policy}, and refer to it as \textbf{``PG-Rank"}. Because PG-Rank is an in-processing algorithm that learns relevance scores, we do not consider this approach in the synthetic experiment case where the relevance scores are already known. All algorithms except Deterministic are stochastic in nature.

Furthermore, we illustrate our algorithm under several choices of $\alpha$ and two different choices of $\beta$ as a scalar multiple of $\alpha$ in Equation \eqref{eq:pl}. Similarly, we vary the one hyperparameter in Scaled Plackett-Luce that multiplies the relevance scores by a constant before sampling from the Plackett-Luce distribution. See the Appendix in Section \ref{appendix:hyperparameter} for the specific choices of hyperparameters for PL-ICFW and Scaled PL for both experiments. 

\subsubsection{Evaluation Metrics}
 We measure performance by looking at the relationship between item-side inequality of exposure and content quality of items recommended to users. We make the normative assumption that viewed recommendations being concentrated to a small number of users is inherently unfair. Therefore, we evaluate system-level inequality by evaluating what percent of viewed recommendations belongs to the top 1\% of users. We refer to this metric as \textbf{``T1PS."} This metric for measuring system-level inequality has been recently proposed in \cite{lazovich2021measuring}  and has been used in the empirical work in \cite{zhu2016attention} for measuring inequality on social media. This metric is related to the Gini coefficient, which has also recently been used to measure inequality in recommender systems \cite{do2022optimizing, do2021two}. We opt to use T1PS since it is more interpretable than the Gini coefficient. We call the performance of the ranking algorithm on the users receiving recommendations the \textbf{``content quality."} We measure content quality by the sum of the ground truth relevance scores of items that have been recommended and viewed by users divided by the number of users $m$. Note that a candidate's ground truth score may contribute to content quality multiple times if that candidate was viewed by multiple users. To be clear, for the German credit experiments, the ground truth relevance score of a candidate is its ground truth binary label. For the synthetic experiments, the relevance score is a continuous number, not a binary label. 
\subsection{Synthetic Data} 
\subsubsection{Set-up} In our synthetic data set-up, there are $n = 2000$ users and $m = 1000$ unique candidates. The candidate set size is $k=40$, and each user views the first $\ell=10$ candidates they are recommended. For each candidate $i$ with candidate score $c_i$, we assign a ground truth relevance score of $r_i := \max\{0, 5 - c_i + x_i\}$ where $x_i \sim N(0,1)$. 

\subsubsection{Results} See Figure \ref{fig:synthetic} for the performance of all the algorithms on the synthetic data set. There are several behaviors in this plot we would expect. First, our algorithm smoothly interpolates between Inverse Weighted and Deterministic. Second, we see the expected behavior that Scaled Plackett-Luce can smoothly interpolate between Randomized and Deterministic. Third, as expected due to our simulation set-up, we see that the Randomized scenario results in the most unequal outcome for producer-side exposure: more than 20\% of viewed recommendations are attributed to the top 1\% of recommended items. Fourth, the Deterministic ranking has the highest content quality as expected since by definition it recommends items with the highest relevance scores.

Interestingly, for the range of hyperparameter choices for Scaled PL\footnote{Recall this hyperparameter is a constant factor such that each score is multiplied by it before sampling from the Plackett-Luce model.}, while content quality monotonically decreases as the hyperparameter decreases, T1PS initially decreases but then gravely increases again, far beyond Deterministic. On the other hand, for our approach when $\beta = \alpha$, while T1PS almost monotonically decreases as $\alpha$ gets smaller, content quality does not and actually increases as $\alpha \rightarrow 0$. We see that for relatively large hyperparameter choices that result in the sampled rankings being close to the Deterministic ranking, both our PL-ICFW method and Scaled PL behave similarly. However, the two methods diverge with relatively small hyperparameter choices, and $\beta = \alpha$ by far outperforms $\beta = .35\alpha$. Our PL-ICFW method under both choices of $\beta$ significantly outperforms Scaled PL since there are several instances where PL-ICFW has the same or better content quality and the same or better T1PS. The minimum T1PS that Scaled PL can achieve is ~6\% whereas our approach achieves ~2\%. We do not expect to achieve the perfect case of T1PS being 1\% since in general it is impossible to achieve equality of exposure among all producers. See Section \ref{sec:impossibility} in the Appendix.

\begin{figure}
\includegraphics[width=.9\textwidth]{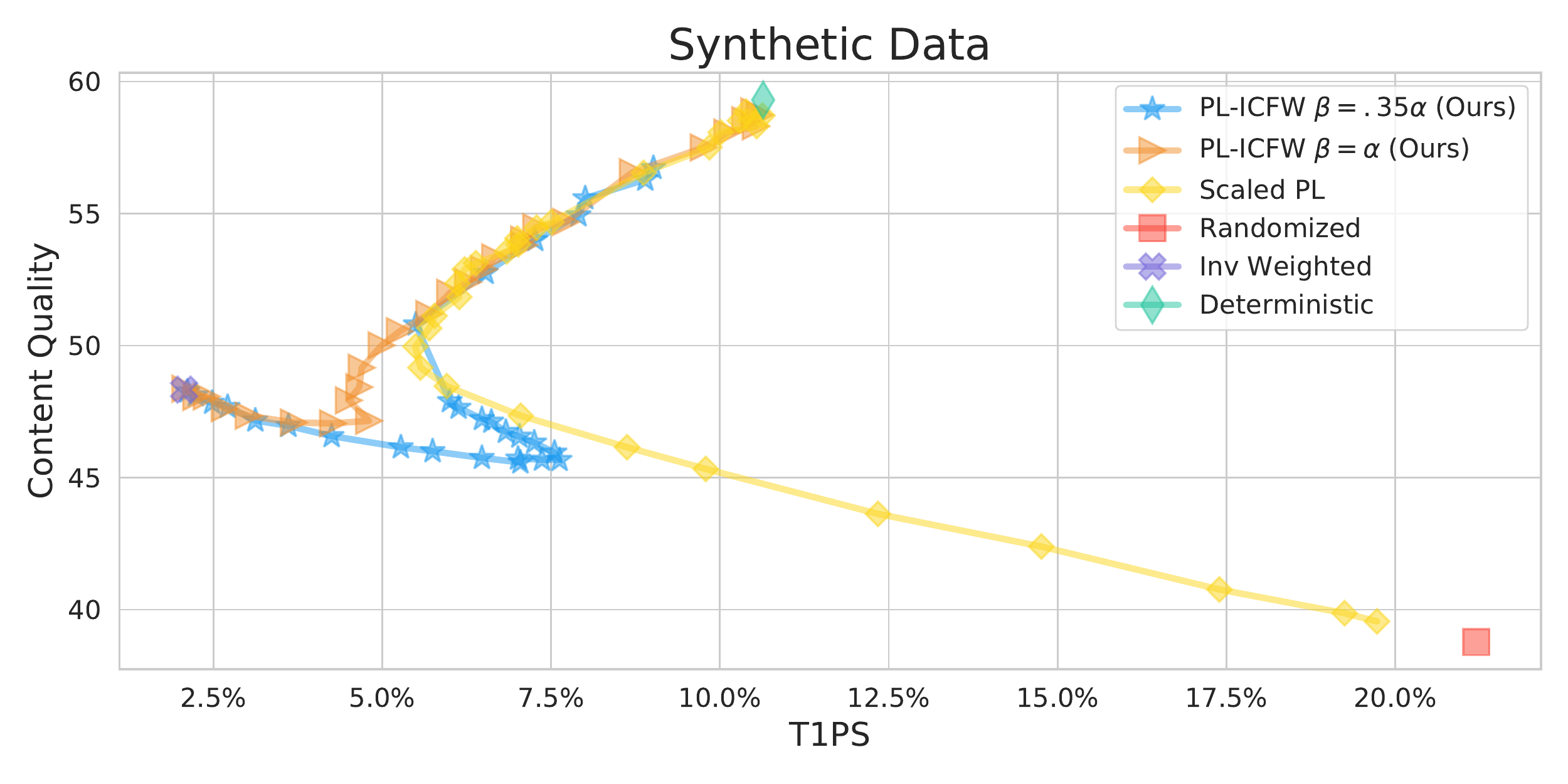}
\caption{Comparison of different algorithms on the synthetic data set measured in terms of inequality to candidate items being recommended and content quality to the users who receive the recommendations. }
\label{fig:synthetic}
\end{figure}

\subsection{German Credit Data}

\subsubsection{Set-up} The German credit data set contains 1000 individuals seeking a loan from a bank. Each are labeled as good or bad credit risk and have features relating to demographic information as well as employment, education, credit history, and so on \cite{Dua:german}. We follow the pre-processing steps of Singh. et. al.\footnote{Using the notebook here: \url{ https://github.com/ashudeep/Fair-PGRank/blob/master/GermanCredit/}} to allow for comparison with the PG-Rank work \cite{singh2019Policy}. Nine features from the dataset are used as inputs to a linear model: age, sex, amount and duration of credit history, job and housing status, savings and checking account balances, and purpose of the loan. The numerical features age, credit amount, and credit duration are all shifted and scaled to a standard normal distribution. The remaining categorical features are encoded as one-hot vectors. Ultimately, this leads to a 29-dimensional input to the linear model. 

We train a linear model with no fairness regularization based on the code in the notebook referenced above to obtain predicted relevance scores of the candidates for all algorithms except PG-Rank. Our algorithm and all the other baseline algorithms except PG-Rank require relevance scores since these algorithms do not learn the scores themselves. We also train several individually fair PG-Rank models based on their code by varying their hyperparameter that controls disparities in exposure at the individual level. See the Appendix in Section \ref{appendix:hyperparameter} for the specific hyperparameter choices.

The ten candidates whose candidate scores are 5 as discussed in Section \ref{sec:candidates} are the ten candidates with the 50-60 lowest learned relevance scores. There are $m = 200$, candidates, $n = 2000$ users who request recommendations, the candidate sets have size $k = 15$, and each user requesting recommendations sees $\ell = 5$ recommendations. 

\subsubsection{Results} 

See Figure \ref{fig:german} for a comparison of the algorithms on the German credit data set. Unsurprisingly, similar to the synthetic data set case, we observe the same four expected behaviors as discussed before. When T1PS is larger than ~16\%, the models all perform similarly with PL-ICFW with $\beta=\alpha$ slightly underperforming the other three algorithms. However, for content quality less than 4.2, our algorithm PL-ICFW under both choices of $\beta$ significantly outperforms the other algorithms since for the same level of content quality, our method has a significantly lower T1PS. All other methods cannot achieve T1PS lower than ~12\% while our method achieves values of T1PS less than ~3\% for the same level of content quality.

\begin{figure}
\includegraphics[width=.9\textwidth]{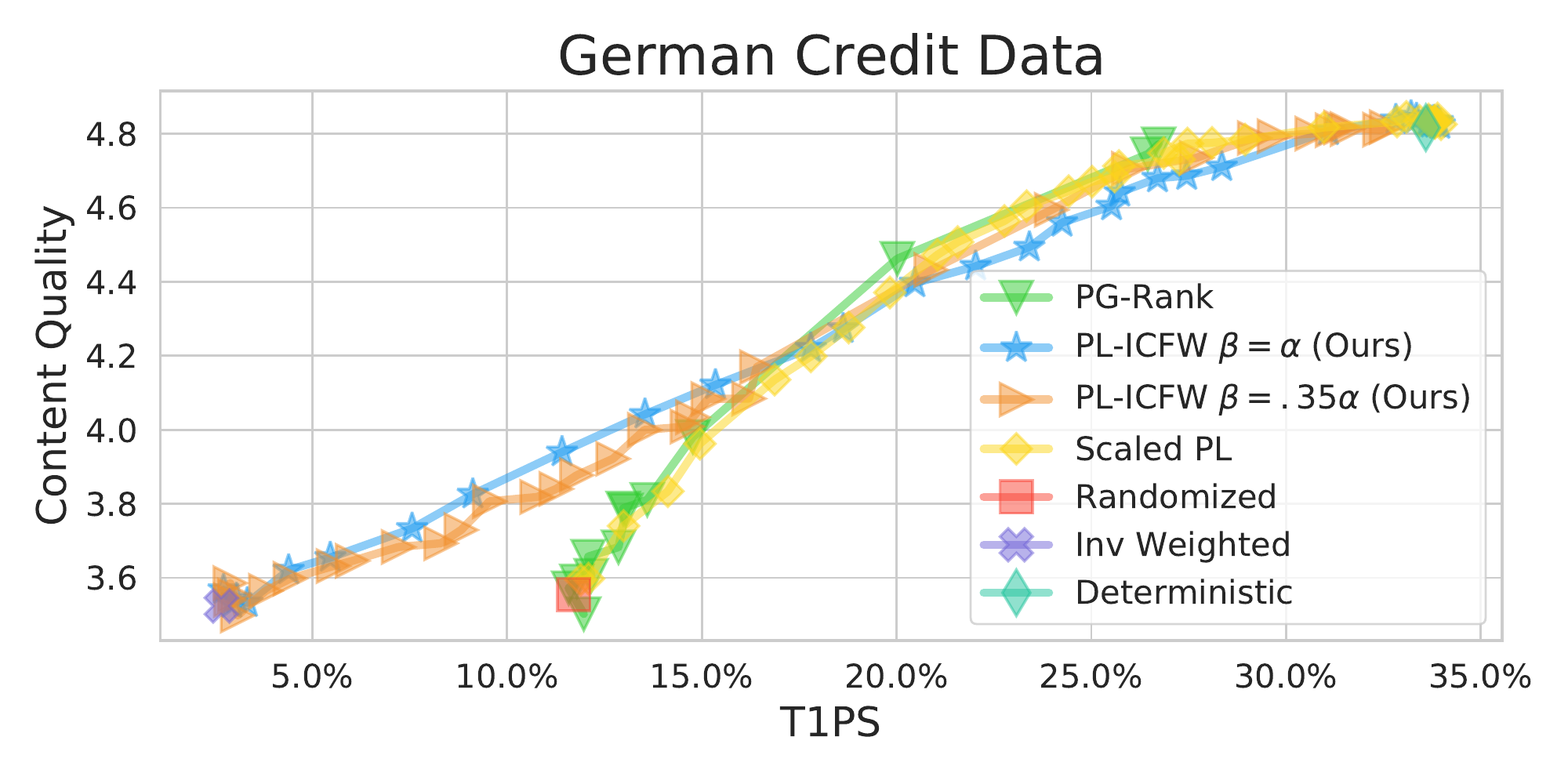}
\caption{Comparison of different algorithms on the German credit data set measured in terms of inequality to candidate items being recommended and content quality to the users who receive the recommendations.}
\label{fig:german}
\end{figure}

\section{Conclusion And Limitations}
\subsection{Limitations}
In many settings, the set of items to rank is continually changing and we may not know the exact frequency with which items will occur in candidate sets. We may need to estimate these frequencies based on a prior time period's data. It is possible that the estimates are sensitive to the choice of time period. Although our post-processing is simple, the optimality of the distribution in Equation \eqref{eq:pl} is unknown. There may be better distributions that trade-off consumer-utility and system level exposure. Additionally, because our algorithm requires hyperparameters and the ranking distribution is a function of the relevance scores, any change to the distribution of the item scores due to a new model will require a new hyperparameter search. This issue is not limited to our approach, but any post-processing approach with hyperparameters, like score boosting approaches \cite{ha2020counterfactual, nandy2021b}. Finally, if the candidate sets contain high levels of inequality, it may be infeasible for any approach that does not directly modify the candidate sets to achieve equitable exposure.

\subsection{Conclusion}
There are several aspects to consider when building responsible recommender systems, like demographic disparities in model performance. Our work specifically focuses on one aspect: inequality of exposure of the items being recommended. We show that even common-sense solutions to ``fair ranking" can behave unexpectedly if they do not take into account candidate set inequality. We showed that even when individual ranking level exposure is completely equal for a given set of candidates, system level exposure measuring the cumulative exposure items receive over all users may not be due to candidate set inequalities. Although others have pointed out the distinction between individual level exposure and system level exposure \cite{sonboli2022multisided, raj2020comparing}, we provide a concrete example of the seemingly paradoxical relationship between system and individual ranking level exposure. We proposed a simple, computationally inexpensive post-processing algorithm that trades-off consumer-utility and producer-side exposure both in cases with that do and do not exhibit high candidate set imbalance. Additionally, our post-processing algorithm provides a baseline for producer-side experimentation, where the effects of algorithmic changes are measured on the producers of content \cite{ha2020counterfactual, nandy2021b}, 
since we can allocate exposure to each item relatively evenly.


\bibliographystyle{ACM-Reference-Format}
\bibliography{AB}

\appendix

\section{Impossibility of Equal Exposure}\label{sec:impossibility}
In this section, we show theoretically why it is impossible to achieve equality of exposure among all producers in general.

\subsection{Notation and Model}
Suppose there are $n$ consumers $\{c_i\}_{i=1}^n$ receiving recommendations of $m$ producers $\{p_i\}_{i=1}^m$. Each consumer $c_i$ has a candidate set of producers $C_i = \{p_{c_{i_j}}\}_{j=1}^{M}$ such that some of these candidates are recommended to the consumer in a ranked list. We assume each producer can show up only once in these sets. We will also make the simplifying assumption that $|C_i| = |C_j| = M$ for all $i,j$. Let $D_i = \{c_{p_{i_j}}\}_{j=1}^{p_{i_N}}$ be the set of $p_{i_N}$ consumers such that there exists $C_{p_{i_j}}$ where producer $p_i \in C_{p_{i_j}}$. Let $I_{c_i}$ be the set of producers that consumer $c_i$ impresses on. Let $E_j = \sum_{c_i \in D_j} \mathds{1}_{p_j \in I_{c_{i}}}$ be the exposure the producer $p_j$ receives measured in number of impressions, where an impression means a consumer saw $p_j$ being recommended to them.

\subsubsection{Model for How Consumers Impress on Recommendations: Subset Selection}
Analyzing how much exposure producers receive requires a model for how consumers impress on recommendations. For each consumer, we pick $k$ candidates to recommend and assume each consumer impresses on all $k$ candidates. The ordering of the candidates is irrelevant since all are impressed on. This model is reasonable when $k$ is small. For example, if a small module recommended three new accounts to follow on a social media platform, it is likely that a user would see all three recommendations. 


\begin{prop}
In general, it is impossible to rank the candidates so that $E_j = E_i$ for all $i,j \in [m]$.
\end{prop}
\begin{proof}
Let $x_{i,p_{i_j}}$ be 1 if for consumer $i$, producer $p_{c_{i_j}} \in C_i$ is selected as one of the $k$ selected producers and 0 otherwise. Clearly under the $k$-subset selection recommendation model, a valid ranking requires $$\sum_{c_{i_j} \in C_i} x_{i,p_{i_j}} = k$$ for all $i \in [n].$ Therefore, we have $n$ linear equations with $Mn$ variables. 

However, in order to satisfy $E_j = E_i$ for all $i, j \in [m]$, we now must satisfy an additional ${m \choose 2}$ linear equations of the same $Mn$ variables since $E_j$ can be written as
$$\sum_{c_i \in D_j} x_{i,p_{i_j}}.$$ Altogether, we must satisfy ${m \choose 2} + n = O(m^2 + n)$ linear equations, in $Mn$ variables. Clearly $M < m$ since there cannot be more candidates than there are producers to recommend. It is also reasonable to assume $m = n$ since every consumer can be a producer and vice versa. Therefore, there are more than $m^2$ linear equations and less than $m^2$ variables. This system of linear equations cannot be satisfied in general.
\end{proof}
Intuitively, it is clear that in general the exposure each producer receives cannot be equal. For example, if one candidate set contains $k+1$ producers that do not appear in any other candidate set, then clearly one producer will necessarily get 0 exposure. Furthermore, if a producer shows up in exactly one candidate set, they can only receive one impression at most.

\section{Hyperparameter Choices in Experiments}\label{appendix:hyperparameter}
For the synthetic experiments, for both PL-ICFW (our method) and Scaled Plackett-Luce, we vary $\alpha$ in our model and the single hyperparameter for PL in $[.01, .025, 1.5,
2, 2.5, 3, 3.5, 4, 4.5]$ and the numbers starting at .05 to 1 in increments of .05. We set $\beta$ to be either $\alpha$ or $.35\alpha.$ 

For the German experiments, for both PL-ICFW (our method) and Scaled Plackett-Luce,  we vary $\alpha$ in our model and the single hyperparameter for PL in $[0.01, 0.025, 1.5, 2]$ as well as the numbers starting at .05 to 1 in increments of .05 and 1 to 7 in increments of $.5$. For PG-Rank, we vary their single hyperparameter in $[.1, 1, 3, 5, 7, 9, 10, 12, 15, 20, 25, 50, 100]$. We set $\beta$ to be either $\alpha$ or $.35\alpha.$ 

\end{document}